\documentclass{article}

\usepackage{PRIMEarxiv}

\usepackage[utf8]{inputenc} 
\usepackage[T1]{fontenc}    
\usepackage{hyperref}       
\usepackage{url}            
\usepackage{booktabs}       
\usepackage{amsfonts}       
\usepackage{nicefrac}       
\usepackage{microtype}      
\usepackage{lipsum}
\usepackage{fancyhdr}       
\usepackage{graphicx}       
\graphicspath{{media/}}     

\usepackage{latexsym}
\usepackage{amssymb}
\usepackage{amsmath}
\usepackage{amsthm}
\usepackage{booktabs}
\usepackage{enumitem}
\usepackage{graphicx}
\usepackage{color}

\usepackage{todonotes}
\usepackage{cleveref} 
\usepackage{mdframed}
\usepackage[most]{tcolorbox}
\usepackage{array}
\usepackage{xspace}
\usepackage{xargs}
\usepackage{thm-restate}

\newtheorem{theorem}{Theorem}

\newtheorem{corollary}[theorem]{Corollary}
\newtheorem{proposition}[theorem]{Proposition}

\newtheorem{claim}[theorem]{Claim} 

\newcommand{\BibTeX}{B\kern-.05em{\sc i\kern-.025em b}\kern-.08em\TeX}

\newcommand{\dist}{\mathrm{dist}}

\newcommand{\Oh}{\mathcal{O}}

\newcommand{\ShortestPath}{\textsc{Shortest Path Most Vital Edges}\xspace}

\newcommand{\tpath}{$T$-path\xspace}
\newcommand{\tpathdel}{\textsc{\tpath-Deletion}\xspace}
\newcommand{\tpathadd}{\textsc{\tpath-Addition}\xspace}
\newcommand{\mtpathdel}{\textsc{\tpath-Deletion with False Promises}\xspace}

\newcommand{\SPMVE}{\textsc{SP-MVE}\xspace}

\newcommand{\fes}{\operatorname{fes}}

\DeclareMathOperator{\operatorClassNP}{{\sf NP}}
\newcommand{\classNP}{\ensuremath{\operatorClassNP}}

\DeclareMathOperator{\operatorClassFPT}{{\sf FPT}\xspace}
\newcommand{\classFPT}{\ensuremath{\operatorClassFPT}\xspace}
\DeclareMathOperator{\operatorClassW}{{\sf W}}
\newcommand{\classW}[1]{\ensuremath{\operatorClassW[#1]}}
\DeclareMathOperator{\operatorClassParaNP}{{\sf Para-NP}\xspace}
\newcommand{\classParaNP}{\ensuremath{\operatorClassParaNP}\xspace}
\DeclareMathOperator{\operatorClassXP}{{\sf X}P\xspace}
\newcommand{\classXP}{\ensuremath{\operatorClassXP}\xspace}

\newtcolorbox{def_box}[1]{enhanced,
  attach boxed title to top left={yshift=-3mm,yshifttext=-1mm, xshift=10mm},
  colback=gray!20!white,colframe=gray!80!black,
  boxrule=1pt,
  title={#1},
  fonttitle=\bfseries,
  coltitle=black,
  boxed title style={size=small,colframe=gray!80!black, colback=gray!40!white, boxrule=1pt},  arc = 2mm}
\newtcolorbox{problem_box}[1]{enhanced,
  attach boxed title to top left={yshift=-3mm,yshifttext=-1mm, xshift=10mm},
  colback=gray!20!white,colframe=gray!80!black,
  boxrule=1pt,
  title={#1},
  coltitle=black,
  boxed title style={size=small,colframe=gray!80!black, colback=gray!40!white, boxrule=1pt},  arc = 2mm}

\newcommandx{\defsimpleproblem}[3][]{
    \begin{problem_box}{#1}
        {\textbf{Input:}} {#2}  \\
        {\textbf{Task:}} {#3}
    \end{problem_box}
}

\newcommandx{\defparproblem}[4][]{
    \begin{problem_box}{#1}
        {\textbf{Input:}} {#2}  \\
        {\textbf{Parameter:}} {#3}  \\
        {\textbf{Task:}} {#4}
    \end{problem_box}
}

\title{How to guide a present-biased agent\\ through prescribed tasks?
}

\author{
  Tatiana Belova\\
  St. Petersburg Department \\
  of Steklov Mathematical Institute\\
  of Russian Academy of Sciences \\
   \And
  Yuriy Dementiev\\
  HSE University\\
  St. Petersburg, Russia\\
  \And
  Fedor Fomin\\
  Department of Informatics,\\
  University of Bergen \\
  Norway\\
  \AND
  Petr Golovach\\
  Department of Informatics,\\
  University of Bergen \\
  Norway\\
  \And
  Artur Ignatiev\\
  HSE University \\
  St. Petersburg, Russia \\
}

\begin{document}
\maketitle

\begin{abstract}
The present bias is a well-documented behavioral trait that significantly influences human decision-making, with present-biased agents often prioritizing immediate rewards over long-term benefits, leading to suboptimal outcomes in various real-world scenarios. Kleinberg and Oren (2014) proposed a popular graph-theoretical model of inconsistent planning to capture the behavior of present-biased agents. In this model, a multi-step project is represented by a weighted directed acyclic task graph, where the agent traverses the graph based on present-biased preferences.

We use the model of Kleinberg and Oren to address the principal-agent problem, where a principal, fully aware of the agent's present bias, aims to modify an existing project by adding or deleting tasks. The challenge is to create a modified project that satisfies two somewhat contradictory conditions. On one hand, the present-biased agent should select specific tasks deemed important by the principal. On the other hand, if the anticipated costs in the modified project become too high for the agent, there is a risk of the agent abandoning the entire project, which is not in the principal's interest.

To tackle this issue, we leverage the tools of parameterized complexity to investigate whether the principal's strategy can be efficiently identified. We provide algorithms and complexity bounds for this problem. 
\end{abstract}


\section{Introduction}

The notion of \emph{present bias} is a standard assumption in behavioral economics used to explain the gap between long-term intention and short-term human decision-making. A present-biased agent prioritizes immediate rewards over long-term benefits, leading to suboptimal outcomes in real-world scenarios. The present bias is one of the reasons for {time-inconsistent behavior} of an agent changing his optimal plans in the short run without new circumstances \cite{o1999doing,thaler2015misbehaving}. Some examples of human time-inconsistent behavior include indulging in unhealthy eating, procrastination on essential tasks and responsibilities, spending on immediate desires instead of saving, addiction abuse despite being aware of the negative consequences or neglecting the immediate efforts in environmental conservation. 

 While originating in behavioral economics, inconsistent planning is related to AI in several ways.
In \emph{Model of Human Behavior}, 
AI systems are often designed to interact with and assist humans. Understanding human behavior, including time inconsistency, is crucial for creating AI systems that can adapt to and predict human actions and preferences. AI models that consider time inconsistency provide more accurate recommendations or assistance \cite{evans2016learning}. In 
\emph{Personalization and Recommendations}, 
 recommendation systems rely on understanding and predicting user preferences. If users exhibit time inconsistency in their preferences, AI systems may need to adapt their recommendations accordingly \cite{DeanM22}. Finally, in 
 \emph{Reinforcement Learning}, agents make decisions to maximize cumulative rewards over time. Time inconsistency can affect an AI agent's ability to make optimal decisions, as it may need to evaluate future rewards and penalties accurately \cite{LesmanaSP22}.

Our work builds on Akerlof's model~\cite{Akerlof91}, in which the \emph{salience factor} causes the agent to prioritize immediate events over the future, with the cost of future tasks assumed to be $1/\beta$ times smaller than their actual costs for some present-bias parameter $\beta < 1$. Even a tiny salience factor could result in significant additional charges for the agent.

Kleinberg and Oren~\cite{KleinbergO14,KleinbergO18} introduced an elegant graph-theoretic model that incorporates the salience factor and scenarios from Akerlof. In this model, an agent traverses from a source $s$ to a target $t$ in a directed edge-weighted graph $G$. We will provide the formal description and begin with an illustrative example.

  \medskip\noindent\textbf{Kleinberg-Oren model example. } 
 Alice is a PhD  student, and she has to accomplish several research projects to obtain her PhD.  After discussing with her advisor Bob, they agree on several possible scenarios, see Fig.~\ref{fig:example}. Every arc of the task graph corresponds to a project, and the cost of an arc is the expected cost required to finish this task. The node $s$ is the starting position of Alice, and the node $t$ is the final node she wants to reach. Thus Alice has three possible options to pursue, corresponding to the three paths in the graph, namely, $P_1=sabct$, $P_2=sadt$, and $P_3=sadet$. She always wants to use the less costly option. To estimate the costs, Alice uses the present-bias parameter $\beta=1/3$---when estimating the cost of a path; she estimates the cost of the first arc correctly. However, she underestimates the costs of all further arcs of the path by factor $\beta$. Thus standing in $s$, Alice estimates the cost of $P_1$ as $6+(2+2+2)/3=8$, the cost of $P_2$ as  $6+(1+6)/3=8\frac{1}{3}$, and $P_3$ as  $6+(1+3+7)/3=9\frac{2}{3}$. She plans to pursue $P_1$. By accomplishing the task $sa$, Alice re-evaluates the remaining costs. The cost of the remaining part of $P_1$ is now  $2+(2+2)/3=3\frac{1}{3}$, which is more than the cost of the remaining part of $P_2$, that is, $1+ 6/3=3$. This impacts Alice's plans and now she decides to follow $P_2$. However, after arriving at $d$, she compares the remaining costs of $P_2$, which is $6$ and $P_3$, which is $5\frac{1}{3}$. Alice changes her plans again and switches to $P_3$.

 \begin{figure}[ht]
 \center{\includegraphics[scale=0.4]{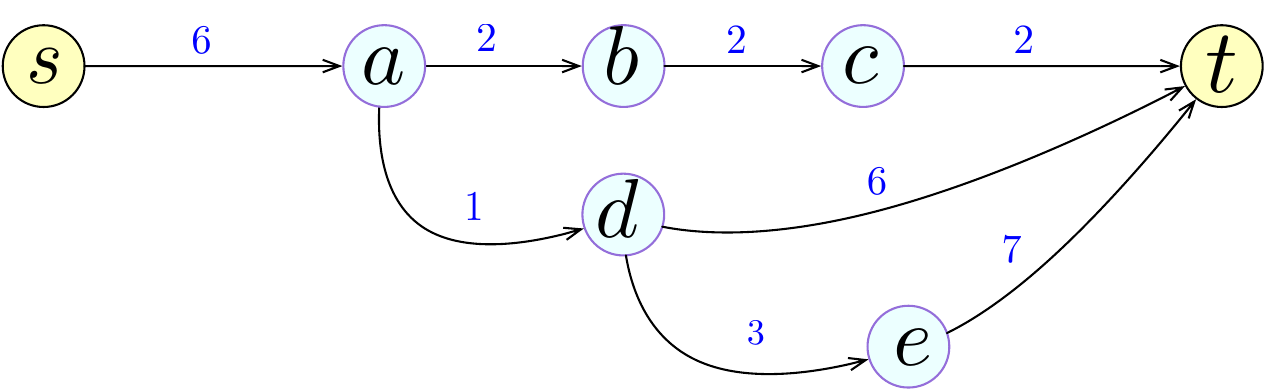}}
 \caption{For $\beta =1/3$, the agent will follow the path
 $sadet$ instead of selecting the shortest path $sabct$.
 }\label{fig:example}
 \end{figure}
 
 In this work, we use the model of Kleinberg and Oren to study a variant of 
 the principal-agent problem, where the principal could reduce the choices 
 to guarantee that the agent will accomplish some selected tasks. 
 We continue with example.
 
   \medskip\noindent\textbf{Motivating by reducing choices.}
   We continue with the example  Fig.~\ref{fig:example}.  To explain the phenomenon of abandonment, Kleinberg and Oren use the reward model. We assume that Alice expects a reward of $r$ for obtaining her PhD.    At every step, she evaluates the cost of completing the path, and if this cost exceeds $\beta  \cdot r$ (reward is also discounted by $\beta$), she abandons the whole project. For this example, we put $r=24$. While Bob, the doctoral advisor of Alice, wants her to finish her study, he has additional interests too. To Bob, the task corresponding to the arc $dt$ is the most exciting part of the whole project. However, if Alice proceeds according to the present bias protocol, she will go through $P_3$ and never accomplish the task so important to Bob. The first thing that comes to Bob's mind---to leave only the tasks of the path $P_2$ available to Alice---does not work.  For Alice standing in $s$, the estimated biased cost of   path $P_2$  is  $8\frac{1}{3}>\beta \cdot r=1/3 \cdot 24=8$. Thus, if Bob leaves $P_2$ as the only choice for Alice, she will abandon her studies. This brings us to the question that is the main motivation for our study. \emph{Is it possible to reduce choices to make both Alice and Bob happy?} That is, Alice will get PhD while working on the tasks that are most interesting to Bob. In our example, the solution is easy---Bob has to delete the task $de$---but in general, as we will see, this question brings interesting algorithmic challenges.  See Fig.~\ref{fig:example2}.
   
  \begin{figure}[ht]
 \center{\includegraphics[scale=0.4]{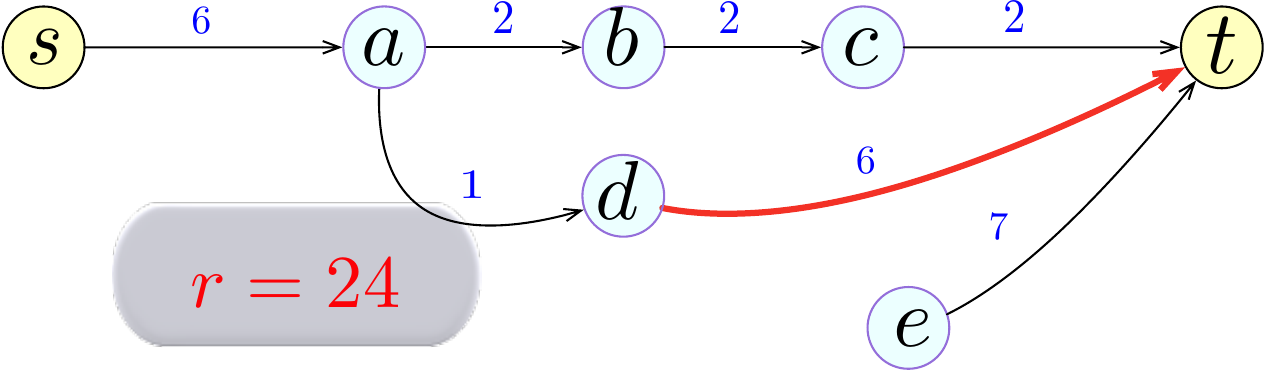}}
 \caption{Let  $P_1=sabct$ and  $P_2=sadt$. For $\beta =1/3$, the agent will follow the path
 $P_2$. Indeed, in node $s$, the estimated cost is $6+1/3(2+2+2)=8$, which is exactly the value $1/3 \cdot r$ of discounted reward, so the agent proceeds to $a$. When standing in  $a$,  the estimated cost of the remaining part of $P_1$ is now  $3\frac{1}{3}$  and of   $P_2$  is $ 3$. Both costs are less than the discounted reward, so the agent follows $P_2$. 
 }\label{fig:example2}
 \end{figure}

We proceed with the formal description of the Kleinberg-Oren's model.

   \begin{def_box}{Kleinberg-Oren's Model~\cite{KleinbergO14}}
    An instance of the {\em time-inconsistent planning model} is a 6-tuple $M = (G, w, s, t,   \beta, r)$ where:
    \begin{itemize}
     \item $G = (V(G), E(G))$ is a directed acyclic $n$-vertex graph called a {\em task graph}. $V(G)$ is a set of elements called {\em vertices}, and $E(G) \subseteq V(G)\times V(G)$ is a set of \emph{arcs} (directed {edges}). 
     Vertices of $G$ represent states of intermediate progress, whereas edges represent possible actions that transition an agent between states.
     \item $w : E(G) \to \mathbb{N}_0$ is a weight function representing the costs of transitions between states. The transition of the agent from state $u$ to state $v$ along arc $uv\in E(G)$ is of cost  $w(uv)$. 
   
     \item The agent starts from the start vertex $s \in V(G)$. 
     \item $t \in V(G)$ is the target vertex.
     \item The rational $\beta  \leq 1$ is the agent's present-bias parameter. 
      \item $r \in \mathbb{Q}_{\geq 0}$ is the reward the agent receives by reaching $t$.
       \end{itemize}
   \end{def_box}
   
 An agent is initially at vertex $s$ and can move along arcs in their designated directions. The agent's task is to reach the target $t$. The agent moves according to the following rule. When standing at a vertex $v$, the agent evaluates (with a present bias) all possible paths from $v$ to $t$.  In particular, a $v$-$t$ path $P \subseteq G$ with edges $e_1, e_2, \ldots,e_p$ is evaluated by the agent standing at $v$ to cost 
 \[\zeta_M(P) = w(e_1) + \beta  \cdot  \sum^{p}_{i=2}w(e_i).\] We refer to this as the {\em perceived} cost of the path $P$. For a vertex $v$, its \emph{perceived cost to the target} is the minimum perceived cost of any path to $t$, \[\zeta_M(v) = \min \{\zeta_M(P) \mid P \text{ is a }v\text{-}t \text{ path}\}.\] We refer to an  $v$-$t$ path $P$ with perceived cost  $\zeta_M(v)$ as to a \emph{perceived path}. If for the agent in vertex $v$ the perceived cost $\zeta_M(v)$ exceeds  $\beta\cdot r$, the value of the reward evaluated in the light of the present bias, the agent abandons the whole project. 
   Thus when in vertex $v$, the agent picks one of the perceived paths\footnote{If there are several paths of minimum perceived cost, we assume that an agent uses a consistent tie-breaking rule, like selecting the node that is earlier in a fixed topological ordering of $G$.} and traverses its first edge, say $vu$.
   After arriving at the new vertex  $u$, the agent computes the perceived cost to the target $\zeta_M(u)$, selects a perceived   $u$-$t$ path, and traverses its first edge. This repeats until the agent either abandons the project or reaches  $t$.
  
 \medskip\noindent\textbf{Guiding through specified arcs.} 
 We are interested in the variant of the principal-agent problem where the principal wants the present-biased agent to perform certain tasks. Using the Kleinberg-Oren model, we model this problem as the following graph modification problem.
 
 For a set of arcs $T\subseteq E(G)$, we say that an $s$-$t$ path $P$ is a \tpath  if $P$ contains all arcs of $T$. 
 Our work addresses the following question. 
 
 \begin{def_box}{}
 For a given set of prescribed tasks $T$, is it possible to modify the time-inconsistent planning model by 
 deleting (or adding) a few tasks such that the present-biased agent will reach $t$ by following a  \tpath? 
 \end{def_box}
 
 Formally, we study the following algorithmic problems. The first problem models the situation when the principal wants to guide the agent through the project by reducing the available options. The second problem models the situation when, instead of reducing the choice, the principal could add more choices from the given tasks. In this case, we assume that we will not create directed cycles when we add arcs.

\defsimpleproblem{\tpathdel}{Time-inconsistent planning model $M = (G, w, s, t,  \beta, r)$, integer $k$ and  a set of arcs $T\subseteq E(G)$.} 
{Find a subset of arcs $D\subseteq E(G)$ of size at most $k$ (or prove that no such set exists), such that after removing $D$ from $M$, the present-biased agent will follow a \tpath.}
We also consider the problem where instead of reducing the choice, the principal can add more choices from a selected family of tasks. In this case, we assume that we will not create directed cycles when we add arcs.
\defsimpleproblem{\tpathadd}{Time-inconsistent planning model $M = (G, w, s, t,   \beta, r)$, integer  $k$,  a set of arcs $T\subseteq E(G)$, and a set of additional weighted arcs $A\subset V \times V$.} {Find a set $S$ of at most $k$ arcs from $A$ (or prove that no such set exists), such that after adding these arcs to $G$ the agent will follow a \tpath.}

\medskip\noindent\textbf{Parameterized complexity.}
Our work extends the current understanding and offers a nuanced perspective on the interplay between computation tractability, graph theory, and decision-making scenarios involving present-biased agents.   
In our algorithmic study of  \tpathdel and \tpathadd, we use the tools of parameterized complexity. We briefly recap the main definitions. 
A \emph{parameterized problem} is a language $Q\subseteq \Sigma^*\times\mathbb{N}$ where $\Sigma^*$ is the set of strings over a finite alphabet $\Sigma$. Respectively, an input  of $Q$ is a pair $(I,k)$ where $I\in \Sigma^*$ and $k\in\mathbb{N}$; $k$ is the \emph{parameter} of  the problem. 
A parameterized problem $Q$ is \emph{fixed-parameter tractable} (\classFPT) if it can be decided whether $(I,k)\in Q$ in time $f(k)\cdot|I|^{\Oh(1)}$ for some function $f$ that depends of the parameter $k$ only.
\classFPT algorithms can be put in contrast with less efficient \classXP algorithms (for slice-wise polynomial), where the running time is of the form $f(k)\cdot |I|^{g(k)}$, for some functions $f,g$.
Respectively, the parameterized complexity class \classFPT is composed of fixed-parameter tractable problems. 
The $\operatorClassW$-hierarchy is a collection of computational complexity classes: we omit the technical
definitions here. The following relation is known amongst the classes in the $\operatorClassW$-hierarchy:
$\classFPT=\classW{0}\subseteq \classW{1}\subseteq \classW{2}\subseteq \cdots \subseteq \classW{P}$. It is widely believed that $\classFPT\neq \classW{1}$, and hence if a
problem is hard for the class $\classW{i}$ (for any $i\geq 1$) then it is considered to be fixed-parameter intractable. For our purposes,
to prove that a problem is   $\classW{1}$-hard it is sufficient to show that an \classFPT algorithm for this problem yields an  \classFPT algorithm for some  $\classW{1}$-complete problem.
We also use notation \classParaNP-hard with parameter $k$ that means \classNP-hard for a constant value of the parameter $k$.
We refer to \cite{CyganFKLMPPS15} for an introduction to parameterized complexity.  

\medskip\noindent\textbf{Our contribution.}
  We start with establishing the hardness of the \tpathdel problem parameterized by $k$. The problem trivially belongs to the class \classXP. An algorithm of running time $|E(G)|^k\cdot  poly(|M|)$ is to try all subsets of at most $k$ arcs and simulate in polynomial time the actions of the agent on the graph, resulting in the removal of each of the subset.  In 
   \Cref{th_w1_k}, we  show that \tpathdel   is \classW1-hard  parameterized by $k$ even when $T$ consists of a single arc. This shows that designing an algorithm of running time $f(k)\cdot  poly(|M|)$ is highly unlikely for any function $f$ of $k$ only. We refine this result in 
    \Cref{cor-para-np} by establishing that \tpathdel problem is \classParaNP-hard 
   with various parameters. In particular, the problem is NP-hard when 
 the maximum cost of the $T$-path is  $ 6$, when  
the reward $r=48$, the model $M$ contains a unique
  $T$-path, or when the input graph
 $G$ has only one heavy arc, and all its other arcs are of weight $1$.
  Thus, 
  \Cref{cor-para-np} refute the existence of parameterized algorithms for many natural parameters of the time-inconsistent model.
  
The intractability results of \tpathdel lead us to contemplate the following question: \emph{although deriving efficient algorithms for general scenarios seems unlikely, could certain structural properties of the instance be algorithmically exploited?}
In other words, while the overall problem may be inherently challenging, there may be specific properties within certain instances that could be leveraged to develop more efficient algorithms.  We introduce two such structural properties, shedding light on potential avenues for algorithmic improvement in the context of \tpathdel.
  
  The first structural property that we exploit algorithmically is the following. Suppose that in the input graph $G$, any path from $s$ to $t$ contains at most $m$ edges. This corresponds to the situation when any sequence of tasks, either taken or anticipated by the agent, contains at most $m$ steps.  In  \Cref{thm-fpt-km},  we give an \classFPT algorithm parameterized by $k$ and  $m$.  Our second parameterization concerns the situation when the underlying undirected graph has a small number of edge-disjoint cycles. Every such cycle could potentially force the agent to change the decision. Thus this parameter is related to the number of nodes where the time-inconsistent agent could change his mind. The main result here is   \Cref{fes-for-TPD}, which establishes the possibility of compressing the instance when the underlying graph of $G$ has a small number of edge-disjoint cycles. More precisely, a feedback edge set of an undirected graph is a set of edges whose removal turns the graph into a forest. Informally,  \Cref{fes-for-TPD} proves that there is a polynomial time algorithm that, for any instance of the problem, constructs an equivalent instance whose size is bounded by a polynomial of the minimum feedback edge set of the underlying undirected graph. 
  In other words, \tpathdel admits a polynomial kernel parameterized by the size of a feedback edge set of the underlying undirected graph. In particular, this implies that the problem is  \classFPT parameterized by the size of a feedback edge set.

  Finally, we provide several algorithmic results for the \tpathadd problem. 
  We consider the case when the set of prescribed arcs $T$ forms a path $P$ containing all vertices of the graph.
  In terms of principal-agent problem, this corresponds to the following interesting scenario. The principal already decided on the sequence of steps the agent should perform. However, in order for the agent to move along this path, the anticipated cost of the proposed path needs to be lowered. Coming back to our example with Alice and Bob, Bob already knows what work Alice has to perform but Alice is too scared by the anticipated amount of time she has to spend on these tasks. Could Bob add some tasks (shortcuts to the path) such that Alice at the end will do all the tasks from $T$?
  As we will see in \Cref{thm-hard-addition}, even in this case, \tpathadd remains intractable. On the positive side, in  \Cref{thm:intersection}, we prove that for a wide class of problems with a well-separable properties of additional tasks, the problem becomes  \classFPT.    

\medskip\noindent\textbf{Related work.}
  The mathematical ideas of present bias go back to the 1930s when Samuelson \cite{Samuelson1937} introduced the discounted-utility model. 
  It has developed into the hyperbolic discounting model, one of the cornerstones of behavioral economics~\cite{Laibson1994,McClure2004}.
 The model of time-inconsistent planning that we adopt for our work is due to Kleinberg and Oren~\cite{KleinbergO14,KleinbergO18}. It could be seen as a special case of the quasi-hyperbolic discounting model (see e.g.~\cite{Laibson1994,McClure2004}), which also generalizes both Samuelson's discounted-continuity model~\cite{Samuelson1937} and Akerlof's salience factor~\cite{Akerlof91}.   
While there is a lot of empirical support for this model, there are also known psychological phenomena about time-inconsistent behavior it does not capture~\cite{Frederick2002}. 

 There is a significant amount of follow-up work on the model of Kleinberg and Oren,   see e.g.~\cite{AAAI22_DFI,FominS20,halpern2023chunking,GravinILP16,KleinbergOR16,KleinbergOR17,meyer2022present}.  In particular, the following two problems are most relevant to our model. 
 
 The first problem is of finding a motivating subgraph. In our model, this corresponds to the situation when the set of prescribed arcs $T$ is empty. Tang et al. \cite{tang2017computational}  show that finding motivating subgraphs is NP-complete. They also investigate a few variations of the problem where intermediate rewards can be placed on vertices.
      Albers and Kraft \cite{Albers2018} independently show that finding a motivating subgraph is NP-complete. Furthermore, they show that the approximation version of the problem (finding the smallest $r$ such that a motivating subgraph exists) cannot be approximated in polynomial time to a ratio of $\sqrt{n}/3$ unless P = NP. Still, a $1 + \sqrt{n}$ -approximation algorithm exists. They also explore another variation of the problem with intermediate rewards. 
     Fomin and Strømme \cite{FominS20} studied the parameterized complexity of computing a simple motivating subgraph. Albers and Kraft \cite{AlbersK17} study a variation on the model where the designer is free to raise arc costs.
     
  The second problem related to our work is the $P$-motivating subgraph problem of    \cite{oren2019principal}. In this variant of the principal-agent problem with a present-biased agent, the principal identifies an $s$-$t$ path $P$ in the task graph $G$. Then the question is whether there is a subgraph of $G$, such that in this subgraph, the agent will follow along $P$. In our model, this corresponds to the situation when the prescribed arcs $T$ form the edge set of $P$. Also, the difference with our model is that Oren and Soker \cite{oren2019principal} look for any $P$-motivating subgraph, while in our model, we are interested in a subgraph from the original graph by a small number of arc deletions/additions. Oren and Soker \cite{oren2019principal} prove that the $P$-motivating subgraph problem is NP-complete even when there are only two different costs of arcs. In the same scenario of two costs,  Oren and Soker gave an algorithm that runs in polynomial time when the number of light arcs in the path $P$ is a constant. 
  
 Finally, in graph algorithms, a prevalent subject of interest revolves around graph modifications, wherein the objective is to alter a graph by modifying adjacencies or deleting vertices to achieve a graph with predefined properties. For comprehensive insights into this topic, we direct readers to surveys such as \cite{burzyn2006np,CrespelleDFG23,natanzon2001complexity}.
Our contribution can be viewed as an augmentation to the existing body of literature within this vibrant research domain.

\section{Motivate by Deletion}\label{sec-deletion}

In this section we study the complexity of the \tpathdel problem. We show that it is \classNP-hard, as well as \classW1-hard parameterized by $k$ and several other parameters that naturally arise in this setting. Also in Theorems 3 and 4 we show that the problem admits an \classFPT algorithm with respect to the structural parameter $\fes$, and also that by adding a new parameter---the maximum edge length of the path, one can obtain an efficient parameterized algorithm. 

To prove hardness, we will reduce the \classNP-hard problem \ShortestPath~\cite{Bazgan19} to our problem. The problem is known to be
\classW1-hard parameterized by $k$ even when the arcs' weights are polynomial in the number of vertices of the input graph
\cite{GolovachT11}. The original formulation of the \SPMVE problem assumes an undirected graph, but all the results are preserved for the case of a directed acyclic graph.

\defsimpleproblem{\ShortestPath (\SPMVE)}{A directed acyclic graph $G = (V, E)$ with positive arcs lengths, two vertices $s, t \in G$, and integers $k, \ell \in \mathrm{N}$.}{Is there an arc subset $S \subseteq E, |S| \leq k$, such that the length of a shortest $s$-$t$ path in $G-S$ is at least $\ell$?}



The following theorem rules out algorithms with a running time of $f(k) \cdot |V(G)|^{O(1)}$ for \tpathdel, for any function $f$ of $k$ only.
\begin{theorem} \label{th_w1_k}
   \tpathdel is \classW1-hard  parameterized by $k$ for any $\beta\leq 1$ even when $T$ consists of a single arc and the weights of arcs are polynomial in $|V(G)|$.
\end{theorem}

\begin{proof}
We construct a parameterized reduction from the \SPMVE problem to the \tpathdel problem.  Let 
 $(G, s, t, k, \ell)$ be an input of  \SPMVE such that weights of  arcs of $G$ are bounded by a polynomial in the number of vertices of the  $G$.  We construct an instance of \tpathdel $M = (G', w, s', t', \beta, r)$, integer $k'$ and a set of arcs $T\subseteq E(G')$ such that in $G'$ at most $k'$ arcs can be removed to motivate the agent to pass along the $T$-path if and only if in $G$  it is possible to remove at most $k$ arcs so that the shortest path between $s$ and $t$ is at least $\ell$.

We construct graph $G'$ from $G$ as follows. We start the construction of $G'$ by multiplying all the arcs' weights of $G$ by $2$. Then we add new vertices $s', v_1, t'$ and arcs with the following weights: $w(s'v_1)=0$, $w(v_1t')=2\ell - 1$, $w(s's) = 0$, and  $w(tt') = 0 $, see  Fig.~\ref{th_w1_k_pic}.  We put one prescribed arc $T = \{s'v_1\}$, parameter $k' = k$, and reward $r = \frac{2\ell}{\beta}$. Finally, we make arc $s's$ and $tt'$ of multiplicity $k+1$. Thus $G'$ has $|V(G)| + 3$ vertices and  $|E(G)|+2k + 4$ arcs.








\begin{figure}[ht]
 \center{\includegraphics[scale=0.4]{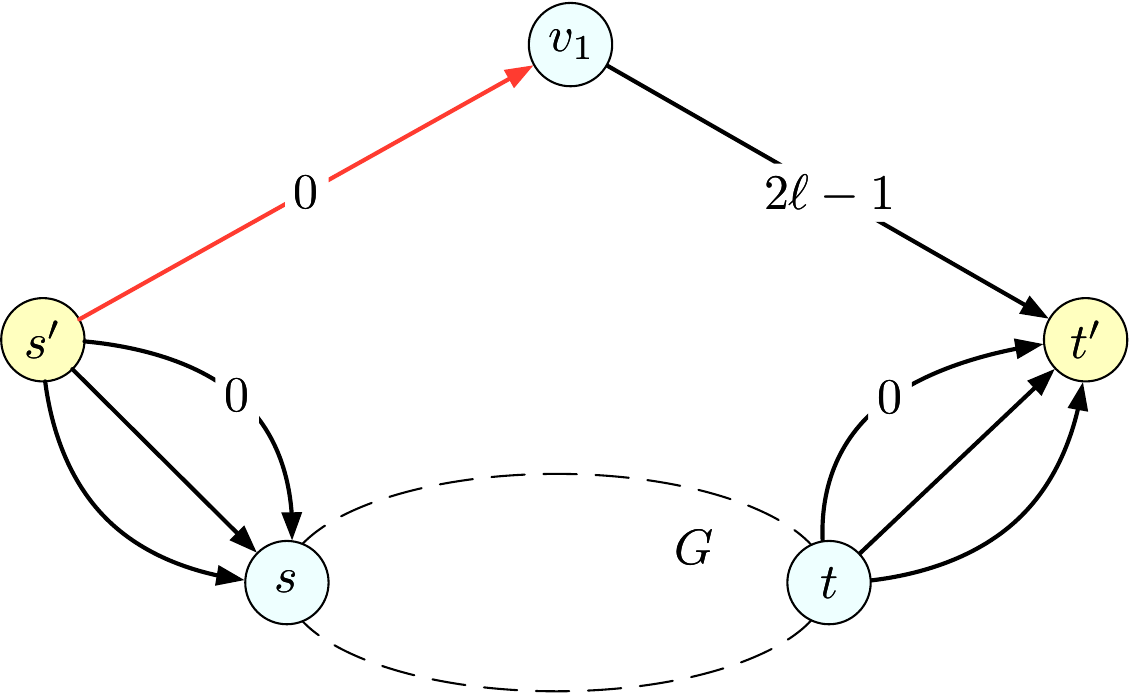}}
 \caption{The construction of the graph $G'$ for \Cref{th_w1_k}.}
 \label{th_w1_k_pic}
\end{figure}

Such a construction could clearly be done in time polynomial in $k$ and $|V(G)|$. Thus, to prove that this is an FPT-reduction, it remains to show that the reduction transforms an instance of \SPMVE into an equivalent instance of \tpathdel. In other words, we have to prove that $(M= (G', w, s', t', \beta, r),k',T)$ is a yes-instance of \tpathdel if and only if $(G, s, t, k, \ell)$ is a yes-instance of \SPMVE.

First, the principal wants the agent to pass through $T = {s'v_1}$ and thus through the path $s'v_1t'$. Hence, none of the arcs of this path could be removed.
Second, in $G'$ arcs $s's$ and $tt'$ are of multiplicity $k+1$, but the principal could remove at most $k' = k$ arcs. Therefore, it is safe to assume that in every solution to \tpathdel, none of these arcs is removed. This allows us to conclude that the only arcs the principal could remove to reach his goal are from $G$. Let $D$ be the set of $k$ arcs deleted by the principal from $G$.

The agent starts at vertex $s'$.  Currently, the perceived cost of the upper path $s'v_1t'$ is  $\beta (2\ell - 1)$. 
If the agent moves to $v_1$, he will follow the path $s'v_1t'$ because the perceived reward $\beta \cdot r$ is always more than the perceived costs along this path at each step. The only reason why the agent decides not to follow this path is that there is another path in $G'-D$ with a smaller estimated cost, which should be at most $\beta(2\ell - 2)$. The arcs $s's$ and $tt'$ are of zero costs, and the principal reaches his goal if and only if graph $G'-D$ has no path from $s$ to $t$ of cost at most  $2\ell - 2$.  Since in $G'$ the weights of the arcs taken from $G$ are twice their original weights in $G$, it means that the agent will move to $v_1$ instead of $s'$  in $G'-D$ (and thus will follow the plans of the principal) if and only if the length of a shortest  $s$-$t$ path in $G-D$ is at least $\ell$.
\end{proof}

\begin{corollary}
The \tpathdel problem is \classW1-hard parameterized by the number of light arcs (set of arcs that have the minimum weight in the instance).
\end{corollary}

\begin{proof}
    We use reduction from \Cref{th_w1_k}. We turn all arcs of cost more than $1$ into a sequence of arcs of weight $1$ in graph $G'$ Fig.~\ref{th_w1_k_pic}. While the size of the graph increases, it is still bounded by a polynomial of the size of the original graph. In the new graph, we have arcs of only two weights, and the number of light arcs equals $2k+3$.
\end{proof}



\begin{corollary}
    The \tpathdel problem parameterized by the parameter pathwidth (and the feedback vertex number) is \classW1-hard. 
\end{corollary}

\begin{proof}
    The \SPMVE problem is \classW1-hard parameterized by pathwidth and maximum degree \cite{BentertHK22}, the problem is also \classW1-hard parameterized by feedback vertex number. In reduction from \Cref{th_w1_k}, the pathwidth of $G'$ has increased by no more than $1$. And the feedback vertex number has increased by no more than $3$.
\end{proof}

\Cref{th_w1_k} can also be generalized to the case of any constant $|T| \geq 1$, we can split the arc $s'v_1$ into path $s'u_1\ldots u_{h}v_1$ with zero arcs where $h$ is a constant and set $T=\{s'u_1, u_1u_2, \ldots, u_hv_1\}$. It means that \tpathdel is \classW1-hard parameterized by $k$ with any constant $|T| \geq 1$. On the other hand, \Cref{th_w1_empty_set} shows that the problem is also \classW1-hard parameterized by $k$ with an empty set $T$.

The lower bound established in \Cref{th_w1_k} immediately questions whether there is a potential for more refined parameterizations to yield parameterized tractability. Unfortunately, the problem is  \classParaNP-hard for many natural parameters like the value of the reward $r$ or the cost of an $T$-path. That is, the problem remains \classNP-hard even when these parameters are constants. We summarize these results in the following theorem.

\begin{theorem}\label{cor-para-np}
The \tpathdel problem remains \classNP-hard even when one of the following conditions holds. 
\begin{enumerate}
    \item The costs of any  $T$-paths in $M$ does not exceed $C\leq 6$.
    \item The reward $r$ is a constant that does not exceed $48$.
    \item There is a unique  $T$-path in $G$.
    \item All arcs in $G$ but one are of weight $1$.
    \item Any path from $s$ to $t$ contains at most $m=8$ arcs.
\end{enumerate}
\end{theorem}

\begin{proof}

    The proof is similar to the \Cref{th_w1_k} except for some parameters. We construct a parameterized reduction from the \SPMVE problem to the \tpathdel problem. Let $(G, s, t, k, \ell)$ be an input of  \SPMVE such that weights of  arcs of $G$ are bounded by a polynomial in the number of vertices of the  $G$.  We construct an instance of \tpathdel $M = (G', w, s', t', \beta, r)$, integer $k'$ and a set of arcs $T\subseteq E(G)$ such that in $G'$ at most $k$ arcs can be removed to motivate the agent to pass along the $T$-path if and only if in $G$  it is possible to remove at most $k$ arcs so that the shortest path between $s$ and $t$ is at least $\ell$.

    We construct graph $G'$ from $G$ as follows. We add new vertices $s', v_1, v_2, t'$ and arcs with the following weights: $w(s'v_1)=\ell/2$, $w(v_1t')=1$, $w(s'v_2) = 1$, $w(v_2s)=\ell$ and  $w(tt') = 1$, see  Fig.~\ref{th1-para-np-pic}.  We put one prescribed arc $T = \{s'v_1\}$, parameter $k' = k$, and reward $r = \frac{\ell}{\beta}$. Finally, we make arc $s'v_2$, $v_2s$ and $tt'$ of multiplicity $k+1$. Thus $G'$ has $|V(G)| + 4$ vertices and  $|E(G)|+3k + 5$ arcs. 









    \begin{figure}[ht]
    \center{\includegraphics[scale=0.4]{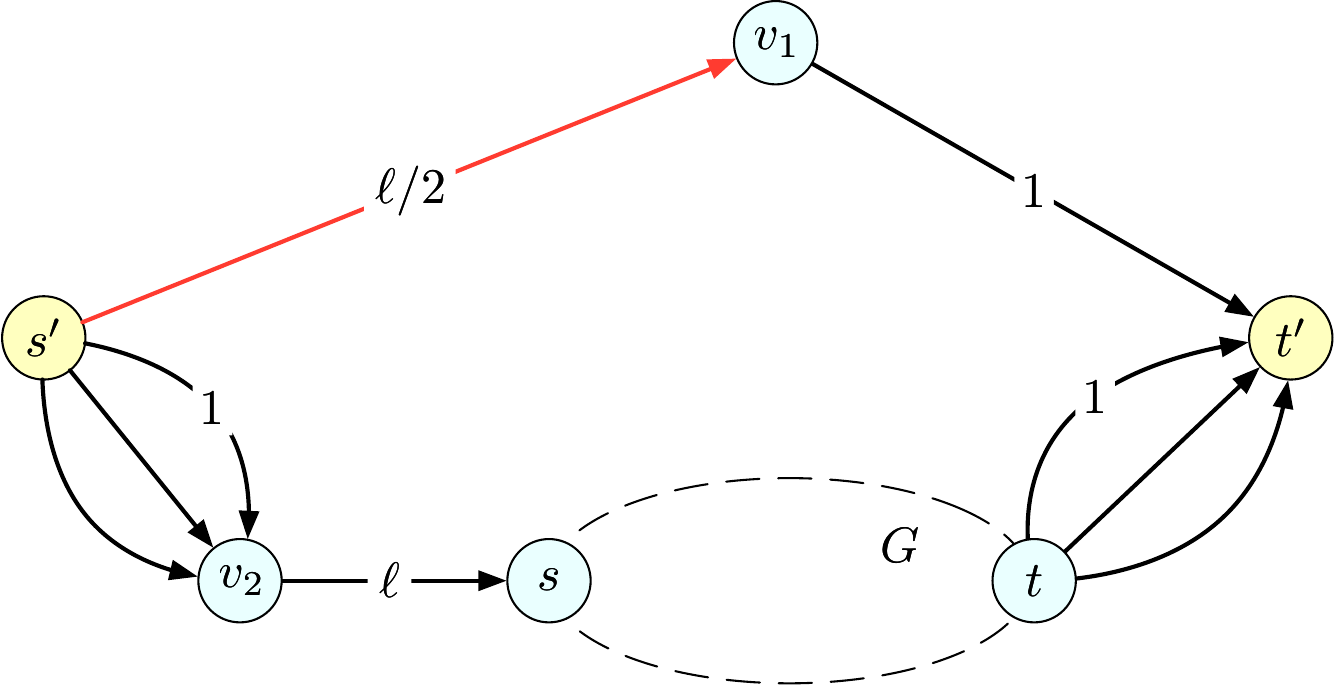}}
    \caption{The construction of the graph $G'$ for Theorem 2.}
	\label{th1-para-np-pic}
    \end{figure}

    Note that the agent will not remain at the vertex $s'$ since the anticipated cost of the path $s'v_1t'$ is less than $\beta r$.
    Let's set the agent's estimates at the vertex $s'$: 
    $$ 1 + \beta \ell + \beta (\ell-1) + \beta < \frac{\ell}{2} + \beta < 1 + \beta \ell + \beta \ell + \beta.$$
    Then $ \dfrac{\ell-2}{4\ell} < \beta < \dfrac{\ell-2}{4\ell-2}$, we set $\beta$ any rational number in this interval.
    
    We now show that the answer to the \SPMVE problem is positive if and only if the answer to the \tpathdel problem is positive. Thus, the agent will choose to go to graph $G$ if and only if there is a path between $s$ and $t$ of length no more than $\ell-1$. Hence, in graph $G'$, at most $k$ edges can be removed to motivate the agent to pass along the $T$-path if and only if in a graph $G$, it is possible to remove at most $k$ edges so that the shortest path between $s$ and $t$ is at least $\ell$.

    Let's look at the complexity of the problem with in terms of natural parameters.
     \begin{enumerate}
        \item It is proven in  \cite{Bazgan19} that \textsc{Shortest Path Most Vital Edges} problem is \classNP-hard with $\ell \geq 9$ for undirected graphs. It can be easily  seen that the result also holds for directed acyclic graphs. 
        Hence, in the reduction we solve a \classNP-hard problem with the value of the parameter $\ell=10$ using the solution of the \tpathdel problem with parameter cost of the $T$-path $P$: $\ell' = \ell / 2 + 1 = 6$.
        \item We know that $\frac{\ell-2}{4\ell} < \beta < \frac{\ell-2}{4\ell-2}$, so $\frac{(4\ell-2) \cdot \ell}{\ell -2} < r < \frac{\ell \cdot 4\ell}{\ell-2}$, since $r = \frac{\ell}{\beta}$. When $\ell = 10$ we have $47\frac{1}{2}<r<50$. Thus, the \tpathdel problem is \classNP-hard with $r = 48$.
        \item In graph $G'$ there is only one $T$-path $sv_1t$.
        \item According to  \cite{Bazgan19},  \textsc{Shortest Path Most Vital Edges} remains \classNP-hard on directed graphs with edges of unit wights. In this way, all arcs in graph $G$ will be of unit weight, and the arc $v_2s$ is split into $\ell$ unit arcs with multiplicities of $k+1$. In the proof of Theorem~2 we have all arcs except for one of weight $1$.
        \item \textsc{Shortest Path Most Vital Edges} problem is \classNP-hard on the graphs in which any path from $s$ to $t$ contains at most $5$ arcs \cite{Bazgan19}. Thus, in graph $G'$, any path from $s'$ to $t'$ contains no more than $8$ arcs.
    \end{enumerate}
    
\end{proof}

\begin{theorem}\label{th_w1_empty_set}
\tpathdel is \classW1-hard  parameterized by $k$ even when $T$ is empty set.
\end{theorem}

\begin{proof}
    We use the same reduction from \Cref{cor-para-np} except for the set $T=\emptyset$. Let's look at the agent's decision at the vertex $v_2$, the agent evaluates any path $v_2s\ldots tt'$ at least as $\ell+\beta$, which is more than expected reward $\beta r = \ell$, so if the agent goes to the vertex $v_2$, then he will never reach the vertex $t'$.
\end{proof}



The lower bounds of  \Cref{th_w1_k} and \Cref{cor-para-np} create an impression that no efficient algorithms for \tpathdel could exist for any reasonable scenario. Despite that, we can identify two interesting parameterizations that make the problem computationally tractable. The first parameter models the natural situation when any sequence of tasks, either taken or anticipated by the agent, contains a bounded number of steps $m$. 
In other words, in this model we assume that  in the input graph $G$, any path from $s$ to $t$ contains at most $m$ edges.
Although our problem is \classParaNP-hard for the parameter $m$ and \classW1-hard parameterized by $k$, our next theorem provides an \classFPT algorithm parameterized by $k$ and  $m$.

\begin{theorem}\label{thm-fpt-km}
    \tpathdel problem is solvable in time  $\Oh(m^{2k}) \cdot poly(|M|)$. 
\end{theorem}

\begin{proof} 
To prove the theorem, we employ the classic technique of parameterized algorithms, namely branching. The idea is to identify a subgraph $H$ of $G$ with at most $m^2$ arcs such that if the principal can motivate the present-biased agent to move over edges of $T$ by removing a set $D$ of at most $k$ arcs, then at least one arc of $D$ should be from $H$.

Consider how the present-biased agent navigates from $s$ to $t$ in graph $G$. If the agent's path includes all arcs from $T$, there is no need for the principal to delete any arc from $G$, so we set $D = \emptyset$. Otherwise, we construct a subgraph $H$ of $G$ as follows.

Let $P_0 = sv_1 v_2 \cdots v_{p}$, $p \leq m$, be the path along which the present-biased agent traverses in $G$ from $s$ to $t$ (perhaps not reaching the vertex $t$ if the agent abandons the project at vertex $v_p$). When standing at a vertex $v_i$, $1\leq i\leq p-1$, the agent evaluates (with a present bias) all possible paths from $v_i$ to $t$. We pick up a path $P_i$ of the minimum perceived cost $\zeta_M(P_i)$ from $v_i$ to $t$. Then we define the graph $H$ as the union of paths $H=\bigcup_{i=0}^{p-1}P_i$. For every $0\leq i\leq p-1$, path $P_i$ has at most $m-i$ arcs. Thus, the number of arcs in subgraph $H$ does not exceed $\sum_{i=0}^m (m -i) \leq m^2$. Since computing the perceived cost of a path could be done in polynomial time \cite{KleinbergO14}, the time required to construct graph $H$ is polynomial in the input size.

Let $D\neq \emptyset$, $|D|\leq k$, be the arcs the principal deletes to achieve his goals. We claim that at least one arc of $D$ is from $H$. Indeed, if this is not the case, then the minimum value $\zeta_M(v_i)$ for each vertex $v_i$ in graph $G-D$ does not change. Hence, if none of the arcs of $H$ are deleted, the agent will traverse $G-D$ along the path $P_0$ and thus will not traverse all arcs from $T$.

This suggests the following branching algorithm. We go through all arcs of $H$. By the above arguments, we know that at least one of the arcs, say $e$, is in $D$. Thus, for the correct guess of the arc $e$, we have that $M = (G, w, s, t, \beta, r)$ with parameter $k$ is a yes-instance if and only if $M' = (G-e, w, s, t, \beta, r)$ with parameter $k-1$ is a yes-instance. In other words, we employ the following branching algorithm:

(i) Compute graph $H$ and branch into $|E(H)|=\Oh (m^2)$ subproblems, corresponding to removing an arc from $G$ and reducing the parameter by $1$.
(ii) Repeat the procedure recursively. That is, in polynomial time, we find a new path of the agent $P'$ in graph $G':=G-e$ and check if it contains all the selected arcs. If yes, then we stop; otherwise, go to step (i).

To analyze the running time of the algorithm, we obtain a branching tree of depth $k$ and arity at most $m^2$, and thus with $\Oh(m^{2k})$ nodes. For each tree node, we compute graph $H$, which is done in time polynomial in $|M|$. Thus, the running time of the algorithm is $\Oh(m^{2k}) \cdot \text{poly}(|M|)$.
\end{proof}

Our second algorithmic result about \tpathdel concerns the limited number of situations when an agent could change a decision. Let us note that the agent could change his mind only when he is on a vertex of some cycle of the underlying undirected graph. The following parameterization concerns the situation when the underlying undirected graph has few edge-disjoint cycles.

A \emph{feedback edge set} of an undirected graph $G$ 
is the set of edges whose removal makes $G$ acyclic.
For a directed graph $G$, we use $\fes(G)$ to denote the minimum size of a feedback edge set of the underlying undirected graph of $G$. Equivalently, $\fes(G)$ is the \emph{cyclomatic} number of the underlying graph.  Note that if $G$ is weakly connected, that is, the underlying graph of $G$ is connected, then $\fes(G)=|E(G)|-|V(G)|+1$.

We consider kernelization for \tpathdel parameterized by $\fes(G)$.
Recall that a \emph{kernelization} algorithm, given an instance $(x,k)$ of some parameterized problem, runs in polynomial time and outputs an equivalent instance $(x',k')$ of the same problem such that $|x'|,k' \leq f(k)$, for some function $f$. This instance $(x',k')$ is called the \emph{kernel}, and function $f$ is called the size of the kernel. Kernelization is one of the fundamental tools for parameterized algorithms, and it is well known that a problem admits a kernel from parameter $k$ if and only if there exists an \classFPT algorithm for it from this parameter. We refer to books \cite{CyganFKLMPPS15,kernelizationbook19} for further expositions of kernelization. In particular, our kernelization algorithms implies that \tpathdel is \classFPT parameterized by $\fes(G)$. 

It is convenient to work with the more general variant of the problem, called \mtpathdel, where promised rewards for distinct vertices may be distinct. Formally, we consider time-inconsistent planning models $M = (G, w, s, t,   \beta, r)$ where $r\colon V(G)\rightarrow \mathbb{Q}_{\geq 0}$. In this variant of the model, the agent occupying a vertex $v$ abandons the project if 
$\zeta_M(v)>\beta\cdot r(v)$. Note that \tpathdel is a special case of \mtpathdel where $r(v)=r$. 

The size of our kernel will depend on $\fes(G)$  only. In particular, it would be independent of the sizes of weights and reward. To obtain such a kernel, we have to compress weights and rewards. 
For this, we use the approach proposed by Etscheid et al.~\cite{DBLP:journals/jcss/EtscheidKMR17} that is based on the result of Frank and Tardos~\cite{DBLP:journals/combinatorica/FrankT87}.

\begin{proposition}[\cite{DBLP:journals/combinatorica/FrankT87}]\label{prop:FT}
 There is an algorithm that, given a vector $\mathsf{w}\in \mathbb{Q}^d$ and an integer~$N$, in polynomial time finds a vector $\overline{\mathsf{w}}\in \mathbb{Z}^d$ with $\|\overline{\mathsf{w}} \|_{\infty}\leq 2^{4d^3}N^{d(d+2)}$ such that 
 $\mathsf{sign}(\mathsf{w}\cdot b)=\mathsf{sign}(\overline{\mathsf{w}}\cdot b)$ for all vectors $b\in\mathbb{Z}^d$ with $\|b\|_1\leq N-1$.
\end{proposition}

\begin{theorem}\label{fes-for-TPD}
   There is a polynomial-time algorithm that, given an instance of  \mtpathdel, outputs an equivalent instance where the graph has at most $8\fes(G) + 3$ vertices and at most $9\fes(G)+2$ arcs. Moreover, if the weights and rewards are rational, and $\beta$ is a rational constant that is not a part of the input, then 
   the \mtpathdel problem admits a polynomial kernel when parameterized by the size of a  feedback edge set of the input graph. 
\end{theorem}


\begin{proof}
    Let $(M, k, T)$ be an instance of \mtpathdel with $M = (G, w, s, t,   \beta, r)$. 
    Let also $f=\fes(G)$.
    We apply the following reduction rules.

    \noindent
    \textbf{Rule 1.} If there is $v \in V(G) \setminus \{s, t\}$ with $d_{in}(v) = 0$ or $d_{out}(v) = 0$, then set $G:=G-v$. Furthermore, if $v$ is incident to an arc from $T$, stop and return a trivial no-instance of \mtpathdel.

    The rule is \emph{safe}, that is, it returns an equivalent instance of the problem because $v$ with  $d_{in}(v) = 0$ or $d_{out}(v) = 0$ cannot be involved in any $s$-$t$ path or agent's evaluation. We apply Rule~1 exhaustively. The next rule is trivially safe. 

    \noindent
    \textbf{Rule 2.} If $t$ is not reachable from $s$ then stop and return a trivial no-instance. 
    
    Notice that if we did not stop after applying the rules then $s$ and $t$ are unique source and target, respectively, of $G$. In particular, for each vertex $v$, $v$ is reachable from $s$, and $t$ is reachable from $v$.  The next rule is crucial for kernelization.

    \noindent
    \textbf{Rule 3.} If $G$ has a path $xyz$ such that $d_{out}(x)=1$ and $d_{in}(y)=d_{out}(y)=1$ then
    \begin{itemize}
        \item delete $y$ and add an arc $xz$,
        \item set $w(xz):=w(xy)+w(yz)$,
        \item if $T\cap\{xy,yz\}\neq\emptyset$ then set $T:=(T\setminus\{xy,yz\})\cup\{xz\}$,
        \item set $r(x):=\min\{r(x)+\frac{1-\beta}{\beta}w(yz),r(y)+\frac{1}{\beta}w(xy)\}$.
    \end{itemize}

    To argue that the rule is safe, assume that the instance $(M', k, T')$ is obtained by the application of the rule from  $(M, k, T)$ and denote by $w'$ and $r'$ the obtained weight and reward functions. We claim that the instances are equivalent.

    For the forward direction, assume that $(M,k,T)$ is a yes-instance. Then there is a set of arcs $D$ of size at most $k$ such that after removing $D$ from $G$, the present-biased agent follows a \tpath $P$. We define $D'=(D\setminus \{xy,yz\})\cup\{xz\}$ if $\{xy,yz\}\cap D\neq \emptyset$, and we set $D'=D$ otherwise. Note that $|D'|\leq |D|\leq k$. We claim that $D'$ is a solution to $(M', k, T')$. The claim is trivial if $P$ does not contain $x$ because, in this case, $xy,yz\notin E(P)$. Assume this is not the case and $x\in V(P)$. Let $Q$ be the $x$-$t$ subpath of $P$. Because $d_{out}(x)=1$ and $d_{in}(y)=d_{out}(y)=1$, $xyz$ is a prefix of $Q$. Because the agent does not abandon the project, $\zeta_M(x)\leq\beta\cdot r(x)$ and $\zeta_M(y)\leq \beta\cdot r(y)$. 
    Suppose that $r'(x)=r(x)+\frac{1-\beta}{\beta}w(yz)$. 
    Then  $\zeta_{M'}(x)=\zeta_M(x)+(1-\beta)w(yz)
    \leq \beta\cdot r'(x)$. If  $r'(x)=r(y)+\frac{1}{\beta}w(yz)$ then 
    $\zeta_{M'}(x)=\zeta_M(y)+w(xy)\leq \beta\cdot r'(x)$.  
    Therefore, the agent occupying $x$ would not abandon the project in the modified graph. Thus, the agent would follow the path $Q'$ obtained from $Q$ by the replacement of $xyz$ by $xz$. This implies that the path $P'$ obtained from $P$ by the replacement of $xyz$ by $xz$ is a $T'$-path in $G'-D'$ in the modified instance, and the agent should follow it. We conclude that  $(M', k, T')$ is a yes-instance.

     For the opposite direction, assume that $(M', k, T')$ is a yes-instance and denote by $D'$ a set of arcs of $G'$ of size at most $k$  such that after removing $D'$ from $G'$, the present-biased agent follows a $T'$-path $P'$. If $xz\in D'$, we set $D=(D'\setminus\{xz\})\cup\{xy\}$, and we set $D=D'$ otherwise. By the definition, $|D|=|D'|=k$. We claim that $D$ is a solution to $(M, k, T)$. Similarly to the proof for the forward direction, the claim is trivial if $P'$ does not contain $x$. Let assume that $x\in V(P')$. Then $xz\in E(P')$. Denote by $Q'$ the $x$-$t$ subpath of $P'$. Since the agent follows $Q'$, 
     $\zeta_{M'}(x)\leq\beta\cdot r'(x)$. Because $r'(x)\leq r(x)+\frac{1-\beta}{\beta}w(yz)$,
     $\zeta_{M}(x)=\zeta_{M'}(x)-(1-\beta)w(yz)\leq \beta\cdot r(x)$. Hence, the agent occupying $x$ in $G$ would not abandon the project and go to $y$. Further, we have that
     $\zeta_M(y)=\zeta_{M'}(x)-w(xy)$. Because $r'(x)\leq r(y)+\frac{1}{\beta}w(xy)$, $\zeta_M(y)\leq \beta\cdot r(y)$. Therefore,the agent occupying $y$ in $G$ would go to $z$.
     We obtain that the agent occupying $x$ in $G$ would follow the path obtained from $Q'$ by replacing of $xz$ by $xyz$. This implies that the path $P$ obtained from $P'$ by the replacement of $xz$ by $xyz$ is a $T$-path in $G-D$, and the agent should follow it. 
     Thus, $(M, k, T)$ is a yes-instance. This concludes the proof that the rule is safe.

     Rule~3 is applied exhaustively whenever possible. Assume from now that Rules~1, 2, and 3 cannot be applied to  $(M, k, T)$.
     Observe that the rules cannot increase the feedback edge set of the underlying graph, that is, $\fes(G)\leq f$. We show the following claim.

     \begin{claim}\label{cl:size}
     $|V(G)|\leq 8\fes(G)+3$ and $|E(G)|\leq 9\fes(G)+2$.
     \end{claim}

    \begin{proof}[Proof of Claim~\ref{cl:size}]
    Denote by $H$ the underlying undirected graph of $G$.
    
    We observe that $H$ has no adjacent vertices of degree two in $V(H)\setminus \{s,t\}$. 
    To see this, assume that $x$ and $y$ are adjacent vertices of degree two distinct from $s$ and $t$. We assume without loss of generality that $xy\in E(G)$. 
    Notice that $d_{in}(x)\neq 0$ in $G$ because of Rule~1. Hence, $d_{out}(x)=1$. Similarly, $d_{out}(y)=1$. Because $G$ is acyclic and $y\neq t$, vertex $y$ has a neighbor $z$ distinct from $x$. However, this means that we would be able to apply Rule~3 for the path $xyz$, contradicting our assumptions that the rules cannot be applied. 
    
    Let $F$ be a set of edges of size $\fes(G)$ such that $R=H-F$ is acyclic. Because $G$ is weakly connected, $R$ is a tree. 
    Denote by $X$ the set containing $s$, $t$, and the endpoints of the edges of $F$. Note that $|X|\leq 2\fes(G)+2$. Observe that all the leaves of $R$ are in $X$ because of Rule~1. It is a folklore observation that a tree with $\ell$ leaves has at most $\ell-2$ vertices of degree at least three. Thus, $R$ has at most $2\fes(G)$ vertices $v\in V(H)\setminus X$ of degree at least three. The degrees of vertices in $V(H)\setminus X$ are the same in $H$ and $R$. Hence, $H$ has at most $2\fes(G)$ vertices of degree at least three outside $X$.    
    By our observation that $H$ has no adjacent vertices of degree two distinct from $s$ and $t$, we obtain that $H$ has at most $4\fes(G)+1$ vertices  $v\in V(H)\setminus X$ of degree two because $R$ is a tree. Therefore, the total number of vertices of $G$ is at most $8\fes(G)+3$. Because $R$ has at most $8\fes(G)+2$ edges, the number of arcs of $G$ is at most $9\fes(G)+2$. This concludes the proof.     
    \end{proof}

    Since Rules~1, 2, and 3 can be applied in polynomial time, \Cref{cl:size} concludes the proof of the first part of the theorem. 

    To show the second claim, assume that the weights and rewards are rational and $\beta=p/q$ is a constant. Consider the vector $\mathsf{w}\in \mathbb{Q}^d$ for $d=|V(G)|+|E(G)|$ whose elements are the values of the reward function $r$ for the vertices of $G$ and the weights of arcs. We define $N=d\max\{p,q\}-1$. Then we apply \Cref{prop:FT}. The algorithm outputs a vector $\overline{\mathsf{w}}\in \mathbb{Z}^d$ and we replace the rewards and the weights by the corresponding values of the elements of $\overline{\mathsf{w}}$. We have that $\mathsf{sign}(\mathsf{w}\cdot b)=\mathsf{sign}(\overline{\mathsf{w}}\cdot b)$ for all vectors $b\in\mathbb{Z}^d$ with $\|b\|_1\leq N-1$. In particular, the equality holds for vectors $b$ whose elements are $0,\pm p, q$. This implies that the replacements of the rewards and weights create an equivalent instance. Because the rewards and weights are upper-bounded by $2^{4d^3}N^{d(d+2)}$ and $d=\Oh(\fes(G))$, we obtain that each numerical parameter can be encoded by a string of length $\Oh(\fes(G)^3)$. We conclude that the algorithm outputs an instance of \mtpathdel of size $\Oh(\fes(G)^4)$. This means that we have a polynomial kernel. This completes the proof.    
\end{proof}

In the second part of \Cref{fes-for-TPD}, we assume that $\beta$ is a rational constant that is not a part of the input. However, it can be observed that the claim holds if $\beta=p/q$ for integers $p,q\leq 2^{\fes(G)^c}$ for some constant $c$. Also, we note that because \tpathdel is \classNP-complete for rational weights and any rational positive constant $\beta<1$, any problem from \classNP{} can be reduced to \tpathdel in polynomial time. This implies the following corollary.

\begin{corollary}
If the weights are rational and $\beta$ is a rational constant which is not a part of the input, then \tpathdel admits a polynomial kernel when parameterized by the size of a  feedback edge set of the input graph. 
\end{corollary}

Also, we can solve \tpathdel in \classFPT time using the algorithm from \Cref{fes-for-TPD}---we reduce an instance of \tpathdel to an equivalent instance of \mtpathdel with a graph of bounded size and guess a solution. 

\begin{corollary}
    \tpathdel is solvable in $2^{\Oh(\fes(G))}\cdot |M|^{\Oh(1)}$ time.
\end{corollary}


\section{Motivate by Addition}
In this section, we show that the \tpathadd problem is computationally hard with respect to the number of edges added even on the simplest type of instances 
when the initial graph is a path whose edges form $T$ and only detours are allowed to be added---edges whose start and end belong to the path. In this case, we assume that all the arcs we add go from left to right.
We will call such inputs a \emph{path with detours}. 
To prove the result, we need the following problem:

\defparproblem{\textsc{Modified $k$-Sum}}{Sets of positive integers $X_1, X_2, \dots, X_k$ and integer $Z$.}{$k$.}{Decide whether there is $x_1\in X_1$, $x_2\in X_2$, $\dots,$ $x_k\in X_k$ such that $x_1+\cdots+x_k= Z$.}

It is known~\cite{AAAI22_DFI} that this problem is \classW1-hard parameterized by $k$.

\begin{theorem}\label{thm-hard-addition}
The \tpathadd problem on the path with detours instances is \classW1-hard parameterized by $k$.
\end{theorem}

\begin{proof}
We construct a parameterized reduction of the \textsc{Modified k-Sum} problem to the \tpathadd problem.
 Let $X_1, X_2, \dots, X_k$ and $Z$ be an instance of the \textsc{Modified k-Sum} problem.  We transform the input such that all elements $ x_j \in X_i$ are  in the interval $[b, 2b]$. For that  we assign $b = \underset{i = 1, \dots, k}{\max}\ \underset{x_j\in X_i}{\max}x_j$,  add $b$ to all elements $x_j$, and set $Z:=Z+kb$.

Then we construct an instance of \tpathadd.
\begin{itemize}
    \item Parameter $k$ is  unchanged.
    \item Graph $G$ is the path on $2k+4$ vertices $v_1 v_2 v_3 \dots v_{2k+4}$. We assign a reward $r$ to the vertex $t = v_{2k+4}$ (see Fig.~\ref{path_reduction}).
    \item We set additional arcs $A = \bigcup_{i=1}^k X_i' \bigcup \{v_2v_4, v_3t\}$, where $X_i'$ is the set of multiple arcs $v_{2i+2} v_{2i+4}$, and the weights of these arcs are numbers from the set $X_i$. The weight of arc $v_2 v_4$ is $1$ and the weight of arc $v_3 t$ is $Z + \frac{1}{\beta} - 1$.
\end{itemize}

   \begin{figure}[ht]
        \center{\includegraphics[scale=0.33]{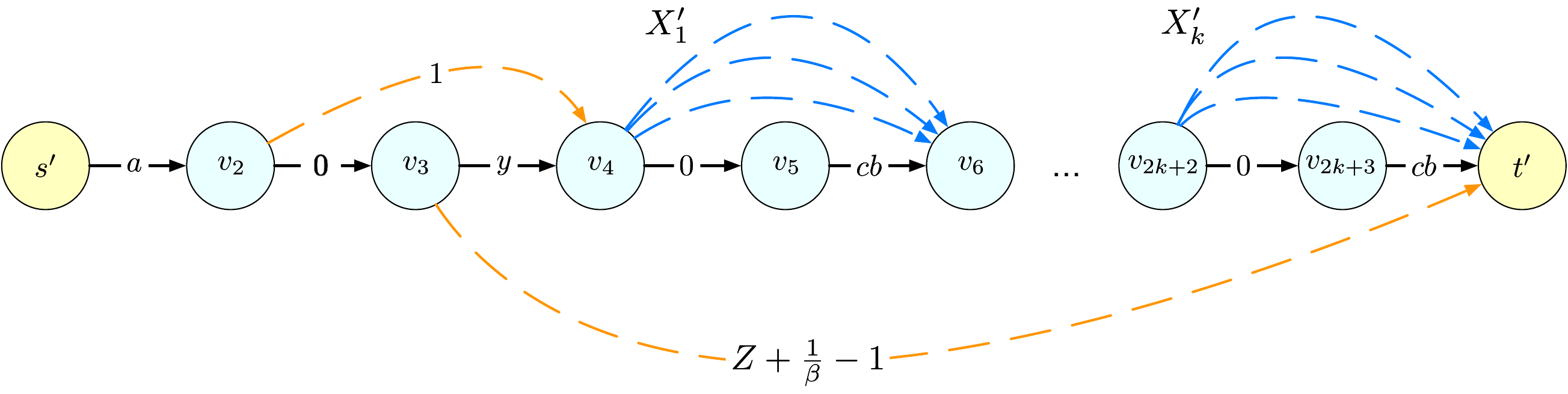}}
		\caption{Construction of the graph for reduction in Theorem~\ref{thm-hard-addition}.}
		\label{path_reduction}
	\end{figure}
 
To make it easier to describe, let us first assume that our input graph $G$ is not just a path but a path with two additional arcs $\{v_2 v_4, v_3 t\}$ (colored in green in the Fig.~\ref{path_reduction}). We will show how to build a reduction for this type of input, and then we will explain why it is possible not to add these two arcs to the path, but to give them in a set of additional arcs $A$.
Note that we take the values of $k, b$ and $Z$ from the input of the \textsc{Modified k-Sum} problem. We will select the parameters $a, c, \beta, r, y$ during the proof.

Initially, we want that the agent in graph $G$ is not motivated  even to leave the start vertex $s=v_1$. For this we put
 \begin{equation}\label{system:eq}\begin{cases}
a + \beta \cdot (1 + k \cdot c \cdot b) > \beta \cdot r, \\
a + \beta \cdot (y + k \cdot c \cdot b) > \beta \cdot r, \\
a + \beta \cdot (Z+\dfrac{1}{\beta} - 1) > \beta \cdot r.
\end{cases}
     \end{equation}
For our purposes, we will consider only the values of $y\geq \frac{1}{\beta} > 1$. Therefore, in \eqref{system:eq}, it is sufficient to satisfy only the first and the third inequalities.
To motivate the agent to move from vertex $s$, it is necessary to add some path from the vertex $v_4$ to $t$ along the arcs from $\{X_i\}_{i=1}^k$.  We denote this  path by $\widehat P$ and its cost, that is the sum of the costs of its arcs,   by $w(\widehat P)$.
For an agent to decide to move from the vertex $s$, it is necessary that his initial estimate be no more than the expected reward, namely:
\begin{equation}
\label{upper}
a+ \beta \cdot (1 + w(\widehat P)) \leq \beta \cdot r.
\end{equation}
But on the other hand, when the agent reaches vertex $v_2$, we need him to continue his path to vertex $v_3$, and not immediately turn into the green arc that leads to vertex $v_4$. From here, we impose the following constraint on the cost of the added path:
\begin{equation}
\label{lower}
1+ \beta \cdot w(\widehat P) > \beta \cdot (Z+\dfrac{1}{\beta} - 1).
\end{equation}
Similarly, the condition that  the agent does not leave  the black path destined for him at vertex $v_3$ is that
\begin{equation}
\label{third}
y+ \beta \cdot w(\widehat P) < Z+\dfrac{1}{\beta} - 1.
\end{equation}

We put $c = 2k$, $a = k \cdot c \cdot b + 1= 2k^2b + 1$, $r = Z + \dfrac{a}{\beta} + 2 - \varepsilon$, where $\varepsilon$---rational number between 0 and 1. Then it is easy to check that all inequalities from system~\eqref{system:eq} are satisfied. We also should select the values of $y$ and $w(\widehat P)$ to satisfy the remaining
 inequalities~\eqref{upper}--\eqref{third}. Note that \eqref{upper}  is equivalent to 
$ w(\widehat P) \leq r - 1 - \dfrac{a}{\beta} = Z+1 - \varepsilon$ and
 \eqref{lower} is equivalent to  $w(\widehat P) > (Z + \dfrac{1}{\beta} - 1) - \dfrac{1}{\beta }= Z-1.$
 
Because the weights of all edges in our construction are integers, we have that under the already existing restrictions, if the agent could go along the black path, then we added the path from $v_4$ to $t$ along the arcs from $\{X_i\}_{i=1}^k$ of cost exactly $Z$.
In each $X_i$ the numbers are in the interval from $b > 0$ to $2b$, hence the value $Z$ does not exceed $2bk$. It follows that if we have found the $v_4$-$t$ path of cost $Z$, then we have added at least one of the proposed arcs in each gadget since the initial arcs have costs $c \cdot b = 2bk \geq Z$.

The inequality \eqref{third} is equivalent to $ y < (Z - \beta \cdot w(\widehat P)) + \dfrac{1}{\beta} - 1.$
Now let us analyze what will happen after the agent reaches vertex $v_4$. It is obvious that at all subsequent vertices, he will be motivated to move on (the estimate of his remaining path will not exceed the expected reward $\beta \cdot r$, since the estimate of the path does not exceed his estimate at vertex $s$, from which he decided to move on).
We only need to make sure that in each gadget the agent does not turn into a blue arc, but continues along the black ones. Since the minimum value of an arc in any $X_i$ is $b$, then the following inequality guarantees that the agent will make the right choice in each gadget:
$$\forall i \geq 4 \quad 0 + \beta \cdot (c \cdot b + w(v_{i+2}\text{-}t \text{ path})) < b + \beta \cdot w(v_{i+2}\text{-}t\text{ path}).$$
$$\beta < \dfrac{1}{c} = \dfrac{1}{2k}.$$

Thus, if we take a present bias coefficient $\beta < 1/2k$ into the input of the \tpathadd problem, then the agent will follow the black path if and only if it is possible to assemble a path of cost exactly $Z$ from the arcs $x_1 \in X_1, x_2 \in X_2, \dots, x_k \in X_k$, or, which is the same when the answer to the \textsc{Modified k-Sum} problem is "yes".

Now we will show how to fine-tune the parameters so that we do not have to force the addition of green arcs to the graph but provide them within the addition set, thus leaving only the path as the input.

In order to add arc $v_2 v_4$, it is necessary that even with the shortest path from vertex $v_3$ to vertex $t$, the agent is not motivated to move from  $s$. In other words, 
\begin{equation*}
\label{y1}
a + \beta \cdot(y + k \cdot b) > \beta \cdot r.
\end{equation*}
$$ y > r - \dfrac{a}{\beta} - kb = Z - kb + 2 - \varepsilon.$$
To add arc $v_3 t$, we need that at vertex $v_2$, the agent turns to the already added arc $v_2 v_4$. That is
\begin{equation*}
\label{y2}
0 + \beta \cdot y > 1.
\end{equation*}
Finally, we set
$$
\begin{cases}
c = 2k, \\
a = 2k^2b+1, \\
\beta < \frac{1}{2k}, \\
r = Z + 2 + \frac{a}{\beta} - \varepsilon, \\
Z - kb + 2 - \varepsilon < y \leq Z - \beta \cdot w(\widehat P) + \frac{1}{\beta} - 1.
\end{cases}
$$
The last double inequality has a solution for $y$, since the left hand side is less than the right hand side:
$$Z - kb + 2- \varepsilon < Z - \beta \cdot w(\widehat P) + \frac{1}{\beta} - 1.$$
$$\beta w(\widehat P) - \frac{1}{\beta} +3 < kb+\varepsilon.$$
We show that starting from some $k$ the desired inequalities hold.
$$ \beta w(\widehat P) - \frac{1}{\beta} +3 < \frac{w(\widehat P)}{2k} - 2k + 3 \leq b - 2k + 3 < kb+ \varepsilon.$$
\end{proof}

By Theorem~\ref{thm-hard-addition}, 
 \tpathadd  is difficult even in the particular case when the set of selected tasks is a Hamiltonian path.
 On the other hand, the problem is easily solvable in time $2^{|A|} n^{\mathcal{O}(1)}$  by going through all potential solutions $S \subseteq A$, $|S| \le k$, and checking in polynomial time whether $S$ is a solution of the problem. Our next theorem generalizes this observation to the situation when the set $A$ has a nice ``separable'' structure. 

Let us start with an example. Let $c \in [n]$, $V_1 = \{v_1, \dots, v_c\}$, $V_2 = \{v_c, \dots, v_n\}$, and let $A$ do not contain any arc $(v_i, v_j)$ such that $i < c < j$.
Let us partition $A$ into two \emph{intersection components} $A_1 = A \cap (V_1 \times V_1)$ and $A_2 \cap (V_2 \times V_2)$, see Fig. \ref{intersections-2}.
We want to show that to solve the problem, we then can solve it on $G[V_1]$ and $G[V_2]$ separately in total time $(2^{|A_1|} + 2^{|A_2|}) \cdot n^{\mathcal{O}(1)}$.

\begin{figure}[ht]
 \center{\includegraphics[scale=0.4]{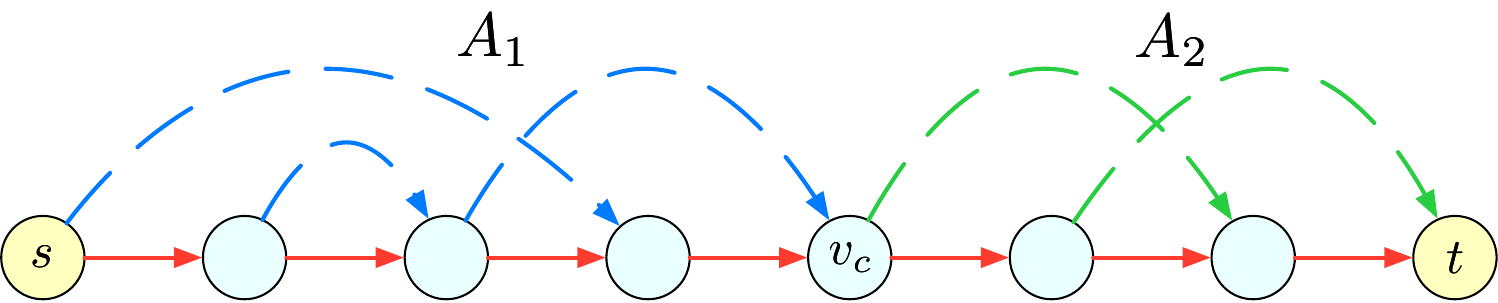}}
 \caption{Intersection components of set $A$ for the \tpathadd problem.}
 \label{intersections-2}
\end{figure}

Let $S \subseteq A$, $S_1 = S \cap A_1$, $S_2 = S \cap A_2$, $k_1 = |S_1|$, $k_2 = |S_2|$. 
Consider the agent's path in $G \cup S$ and also divide it into two parts going through $V_1$ and $V_2$, respectively.
Notice that in case the agent gets to $V_2$, the second part of the path is exactly the agent's path in $G[V_2] \cup S_2$.
Now let us consider the first part of the path.
Notice that for every vertex $v_i \in V_1$, any perceived path induces one of the shortest paths in $G[V_2] \cup S_2$.
That means that the agent's decisions depend only on $G[V_1] \cup S_1$ and $\dist_{G[V_2] \cup S_2}(v_c, v_n)$, where by $\dist_G(s, t)$ we denote weight of the shortest path in $G$ from $s$ to $t$. 
Moreover, $\dist_{G[V_2] \cup S_2}(v_c, v_n)$ takes part only in comparison of a perceived cost with $\beta \cdot r$ that can be replaced with comparison of the perceived cost in $G[V_1] \cup S_1$ with $\beta (r - \dist_{G[V_2] \cup S_2}(v_c, v_n))$, so the first part of the agent's path is exactly the agent's path in $G[V_1] \cup S_1$ with reward $r - \dist_{G[V_2] \cup S_2}(v_c, v_n)$.

Hence, there exists a solution $S$ for the initial problem if and only if there exist $k_1, k_2 : k_1 + k_2 \le k$, $S_1 \subseteq A_1$ of size $k_1$ and $S_2 \subseteq A_2$ of size $k_2$ such that in $G[V_2] \cup S_2$ the agent follows path $v_c \dots v_n$ with reward $r$, and in $G[V_1] \cup S_1$ the agent follows path $v_1 \dots v_c$ with reward $r - \dist_{G[V_2] \cup S_2}(v_c, v_n)$. 
We also notice that for every such $(k_1, k_2, S_1, S_2)$, the above also holds for any $(k_1, k_2, S_1, S_2')$ where $S_2' \subseteq A_2$ of size $k_2$ such that the agent takes path $v_c, \dots, v_n$ in $G[V_2] \cup S_2'$, and $\dist_{G[V_2] \cup S_2'}(v_c, v_n) \le \dist_{G[V_2] \cup S_2}(v_c, v_n)$. 
That means, that it is
sufficient to consider only $S_2$ that minimizes $\dist_{G[V_2] \cup S_2}(v_c, v_n)$. 

Now we can solve \tpathadd in time $(2^{|A_1|}+2^{|A_2|}) \cdot n^{\mathcal{O}(1)}$ in the following way.
For every $k_2 \le k$ we compute $d[k_2] = \{\dist_{G[V_2] \cup S_2}(v_c, v_n) \mid S_2 \subseteq A_2, |S_2| \le k_2$, the agent follows path $v_c \dots v_n$ in $G[V_2] \cup S_2$ with reward $r$\}. That can be done in time $2^{|A_2|} n^{\mathcal{O}(1)}$ by going through all $S_2 \subseteq A_2$.
Then, for every $k_1 \le k$ we go through all $S_1 \subseteq A_1$, $|S_1| = k_1$ and check whether the agent follows path $v_1 \cdots v_c$ in $G[V_1] \cup S_1$ with reward $r - d[k - k_1]$. That can be done in time $2^{|A_1|} n^{\mathcal{O}(1)}$.

Let us now generalize the result. 
Let $1 = c_1 < c_2 <\cdots < c_{m + 1} = n$ be indices such that for every $2 \le \ell \le m$ there is no edge $(v_i, v_j)$ such that $i < c_{\ell} < j$. 
For every $1 \le \ell \le m$, let $V_{\ell} = \{v_{c_{\ell}}, \dots, v_{c_{\ell + 1}}\}$, and let us partition $A$ into intersection components $A_{\ell} = A \cap (V_{\ell} \times V_{\ell})$. 
Then we show that the following theorem holds.

\begin{theorem}\label{thm:intersection}
    The \tpathadd problem on paths with detours can be solved in time $2^{\tau} n^{\mathcal{O}(1)}$, where $\tau$ is the size of the maximum intersection component of $A$.
\end{theorem}

\begin{proof}
    For every $1 \le \ell \le m$, let $A_{\ell}^{\cup} = A_{\ell} \cup \dots \cup A_{m}$ and let
    $d[\ell][\kappa] = \min \{\dist_{G[A_{\ell}^{\cup}] \cup S}(v_{c_{\ell}}, v_{n}) \mid S \subseteq A_{\ell}^{\cup}, |S| \le \kappa, \text{the agent follows path $v_{c_{\ell}} \dots v_n$ with reward $r$}\}$.
    If there is no such $S$ then $d[\ell][\kappa]$ is assigned to $+\infty$.
    Then, to solve the problem it is sufficient to compute all $d[\ell][\kappa]$ and check whether $d[1][k]$ is not equal to $+\infty$.

    We compute $d[\ell][\kappa]$ using dynamic programming technique. Using the same argument as in case of $m = 2$ above, we notice that $d[\ell][\kappa] = \min \{\dist_{G[A_{\ell}] \cup S}(v_{c_{\ell}}, v_{c_{\ell + 1}}) + d[\ell + 1][\kappa - \kappa'] \mid S \subseteq A_{\ell}, |S| \le \kappa' \le \kappa$, the agent follows path $v_{c_{\ell}} \dots v_{c_{\ell + 1}}$ in $G[A_{\ell}] \cup S$ with reward $r - d[\ell + 1][\kappa - \kappa'] \}$.
    We start with $\ell = m$ and iteratively decrease it after exhausting all possible $\kappa \le k$. Then, we can compute all $d[\ell][\kappa]$ and solve the problem in time $2^{\tau} n^{\mathcal{O}(1)}$.
\end{proof}

\section{Conclusion}
In this work, we use the graph-theoretical model of Kleinberg and Oren to introduce  the principal-agent problem, where the principal could reduce the choices 
 to guarantee that the agent will accomplish some selected tasks. 
We conclude with directions for further research and some concrete open problems. 
While we consider only the scenario of deleting and adding arcs, several other natural models would be exciting to explore. The process of adding and deleting arcs could be simulated by changing the weights of the arcs. This more general model, where the principal could change the weights of the arcs in order to motivate the agent, is much more algorithmically challenging. 
Another attractive model is where the principal motivates the agent by putting rewards for accomplishing some intermediate tasks like in   \cite{AlbersK17}. 

As a concrete open question, for the \tpathdel problem, we obtained a kernel whose size is polynomial in the size of a  feedback edge set of $G$. We do not know if a kernel whose size is bounded by a size (even exponential) of a vertex cover of $G$ exists. 

\section*{Acknowledgments}
\emph{Artur Ignatiev and Yuriy Dementiev and Tatiana Belova}: Research is supported by the grant \textnumero075-15-2022-289 for creation and development of Euler International Mathematical Institute.

\noindent
\emph{Tatiana Belova}: Research is supported by the Foundation for the Advancement of Theoretical Physics and Mathematics ``BASIS''.

\noindent
\emph{Fedor Fomin}:
The research leading to these results has been supported by Trond Mohn forskningsstiftelse (grant no. TMS2023TMT01).

\noindent
\emph{Petr Golovach}:
The research leading to these results has been supported by the Research Council of Norway (grant no. 314528).

\bibliographystyle{unsrt}  
\bibliography{main}

\begin{thebibliography}{10}

\bibitem{o1999doing}
Ted O'Donoghue and Matthew Rabin.
\newblock Doing it now or later.
\newblock {\em American economic review}, 89(1):103--124, 1999.

\bibitem{thaler2015misbehaving}
Richard~H Thaler and LJ~Ganser.
\newblock Misbehaving: The making of behavioral economics.
\newblock 2015.

\bibitem{evans2016learning}
Owain Evans, Andreas Stuhlm{\"u}ller, and Noah Goodman.
\newblock Learning the preferences of ignorant, inconsistent agents.
\newblock In {\em Proceedings of the AAAI Conference on Artificial Intelligence
  (AAAI)}, volume~30, 2016.

\bibitem{DeanM22}
Sarah Dean and Jamie Morgenstern.
\newblock Preference dynamics under personalized recommendations.
\newblock In {\em {EC} '22: The 23rd {ACM} Conference on Economics and
  Computation, Boulder, CO, USA, July 11 - 15, 2022}, pages 795--816. {ACM},
  2022.

\bibitem{LesmanaSP22}
Nixie~S. Lesmana, Huangyuan Su, and Chi~Seng Pun.
\newblock Reinventing policy iteration under time inconsistency.
\newblock {\em Trans. Mach. Learn. Res.}, 2022, 2022.

\bibitem{Akerlof91}
George~A. Akerlof.
\newblock Procrastination and obedience.
\newblock {\em American Economic Review: Papers and Proceedings}, 81(2):1--19,
  1991.

\bibitem{KleinbergO14}
Jon~M. Kleinberg and Sigal Oren.
\newblock Time-inconsistent planning: a computational problem in behavioral
  economics.
\newblock In {\em {ACM} Conference on Economics and Computation (EC)}, pages
  547--564, 2014.

\bibitem{KleinbergO18}
Jon~M. Kleinberg and Sigal Oren.
\newblock Time-inconsistent planning: a computational problem in behavioral
  economics.
\newblock {\em Commun. {ACM}}, 61(3):99--107, 2018.

\bibitem{CyganFKLMPPS15}
Marek Cygan, Fedor~V. Fomin, Lukasz Kowalik, Daniel Lokshtanov, D{\'{a}}niel
  Marx, Marcin Pilipczuk, Michal Pilipczuk, and Saket Saurabh.
\newblock {\em Parameterized Algorithms}.
\newblock Springer, 2015.

\bibitem{Samuelson1937}
Paul~A. Samuelson.
\newblock A note on measurement of utility.
\newblock {\em The Review of Economic Studies}, 4(2):155--161, 02 1937.

\bibitem{Laibson1994}
David~I. Laibson.
\newblock {\em Hyperbolic Discounting and Consumption}.
\newblock PhD thesis, Massachusetts Institute of Technology, Department of
  Economics, 1994.

\bibitem{McClure2004}
Samuel~M. McClure, David~I. Laibson, George Loewenstein, and Jonathan~D. Cohen.
\newblock Separate neural systems value immediate and delayed monetary rewards.
\newblock {\em Science}, 306(5695):503--507, 2004.

\bibitem{Frederick2002}
Shane Frederick, George Loewenstein, and Ted O'Donoghue.
\newblock Time discounting and time preference: A critical review.
\newblock {\em Journal of Economic Literature}, 40(2):351--401, 2002.

\bibitem{AAAI22_DFI}
Yuriy Dementiev, Fedor Fomin, and Artur Ignatiev.
\newblock Inconsistent planning: When in doubt, toss a coin!
\newblock {\em Proceedings of the 36th AAAI Conference on Artificial
  Intelligence (AAAI)}, 36(9):9724--9731, Jun. 2022.

\bibitem{FominS20}
Fedor~V. Fomin and Torstein J.~F. Str{\o}mme.
\newblock Time-inconsistent planning: Simple motivation is hard to find.
\newblock In {\em Proceeding of the 34th {AAAI} Conference on Artificial
  Intelligence (AAAI)}, pages 9843--9850. {AAAI} Press, 2020.

\bibitem{halpern2023chunking}
Joseph~Y Halpern and Aditya Saraf.
\newblock Chunking tasks for present-biased agents.
\newblock In {\em Proceedings of the 24th ACM Conference on Economics and
  Computation (EC)}, pages 853--884, 2023.

\bibitem{GravinILP16}
Nick Gravin, Nicole Immorlica, Brendan Lucier, and Emmanouil Pountourakis.
\newblock Procrastination with variable present bias.
\newblock In {\em {ACM} Conference on Economics and Computation (EC)}, page
  361, 2016.

\bibitem{KleinbergOR16}
Jon~M. Kleinberg, Sigal Oren, and Manish Raghavan.
\newblock Planning problems for sophisticated agents with present bias.
\newblock In {\em {ACM} Conference on Economics and Computation (EC)}, pages
  343--360, 2016.

\bibitem{KleinbergOR17}
Jon~M. Kleinberg, Sigal Oren, and Manish Raghavan.
\newblock Planning with multiple biases.
\newblock In {\em {ACM} Conference on Economics and Computation (EC)}, pages
  567--584, 2017.

\bibitem{meyer2022present}
Seth~A Meyer, Jessica Pomplun, and Joshua Schill.
\newblock Present bias in partially sophisticated and assisted agents.
\newblock {\em Mathematical Social Sciences}, 118:36--47, 2022.

\bibitem{tang2017computational}
Pingzhong Tang, Yifeng Teng, Zihe Wang, Shenke Xiao, and Yichong Xu.
\newblock Computational issues in time-inconsistent planning.
\newblock In {\em Proceedings of the 31st Conference on Artificial Intelligence
  (AAAI)}. {AAAI} Press, 2017.

\bibitem{Albers2018}
Susanne Albers and Dennis Kraft.
\newblock Motivating time-inconsistent agents: {A} computational approach.
\newblock {\em Theory Comput. Syst.}, 63(3):466--487, 2019.

\bibitem{AlbersK17}
Susanne Albers and Dennis Kraft.
\newblock On the value of penalties in time-inconsistent planning.
\newblock In {\em 44th International Colloquium on Automata, Languages, and
  Programming (ICALP)}, pages 10:1--10:12, 2017.

\bibitem{oren2019principal}
Sigal Oren and Dolav Soker.
\newblock Principal-agent problems with present-biased agents.
\newblock In {\em Proceedings of the 12th International Symposium on
  Algorithmic Game Theory (SAGT)}, pages 237--251. Springer, 2019.

\bibitem{burzyn2006np}
Pablo Burzyn, Flavia Bonomo, and Guillermo Dur{\'a}n.
\newblock {NP}-completeness results for edge modification problems.
\newblock {\em Discrete Applied Mathematics}, 154(13):1824--1844, 2006.

\bibitem{CrespelleDFG23}
Christophe Crespelle, P{\aa}l~Gr{\o}n{\aa}s Drange, Fedor~V. Fomin, and Petr~A.
  Golovach.
\newblock A survey of parameterized algorithms and the complexity of edge
  modification.
\newblock {\em Comput. Sci. Rev.}, 48:100556, 2023.

\bibitem{natanzon2001complexity}
Assaf Natanzon, Ron Shamir, and Roded Sharan.
\newblock Complexity classification of some edge modification problems.
\newblock {\em Discrete Applied Mathematics}, 113(1):109--128, 2001.

\bibitem{Bazgan19}
Cristina Bazgan, Till Fluschnik, Andr{\'e} Nichterlein, Rolf Niedermeier, and
  Maximilian Stahlberg.
\newblock A more fine-grained complexity analysis of finding the most vital
  edges for undirected shortest paths.
\newblock {\em Networks}, 73(1):23--37, 2019.

\bibitem{GolovachT11}
Petr~A. Golovach and Dimitrios~M. Thilikos.
\newblock Paths of bounded length and their cuts: Parameterized complexity and
  algorithms.
\newblock {\em Discret. Optim.}, 8(1):72--86, 2011.

\bibitem{BentertHK22}
Matthias Bentert, Klaus Heeger, and Dusan Knop.
\newblock Length-bounded cuts: Proper interval graphs and structural
  parameters.
\newblock {\em J. Comput. Syst. Sci.}, 126:21--43, 2022.

\bibitem{kernelizationbook19}
Fedor~V. Fomin, Daniel Lokshtanov, Saket Saurabh, and Meirav Zehavi.
\newblock {\em Kernelization. Theory of Parameterized Preprocessing}.
\newblock Cambridge University Press, 2019.

\bibitem{DBLP:journals/jcss/EtscheidKMR17}
Michael Etscheid, Stefan Kratsch, Matthias Mnich, and Heiko R{\"{o}}glin.
\newblock Polynomial kernels for weighted problems.
\newblock {\em J. Comput. Syst. Sci.}, 84:1--10, 2017.

\bibitem{DBLP:journals/combinatorica/FrankT87}
Andr{\'{a}}s Frank and {\'{E}}va Tardos.
\newblock An application of simultaneous diophantine approximation in
  combinatorial optimization.
\newblock {\em Comb.}, 7(1):49--65, 1987.

\end{thebibliography}

\end{document}